\documentclass[conference]{IEEEtran}  
\IEEEoverridecommandlockouts

\usepackage{cite}
\usepackage{amsmath,amssymb,amsfonts}
\usepackage{amsthm}
\usepackage{algorithm,algorithmic}
\makeatletter
\renewcommand\fs@ruled{\def\@fs@cfont{\bfseries}\let\@fs@capt\floatc@ruled
  \def\@fs@pre{\kern0.06in\hrule height.8pt depth0pt \kern2pt}%
  \def\@fs@post{\kern2pt\hrule\relax}%
  \def\@fs@mid{\kern2pt\hrule\kern2pt}%
  \let\@fs@iftopcapt\iftrue}
\restylefloat{algorithm}
\makeatother
\usepackage[caption=false,font=small,labelfont=sf,textfont=sf]{subfig}
\usepackage{graphicx}
\usepackage{textcomp}
\usepackage{xcolor}
\usepackage{booktabs}
\def\BibTeX{{\rm B\kern-.05em{\sc i\kern-.025em b}\kern-.08em
    T\kern-.1667em\lower.7ex\hbox{E}\kern-.125emX}}

\newcommand{\subc}{{k}{\color{black}}}

\newcommand{\bandK}{{K}{\color{black}}}
\newcommand{\timeT}{{T}{\color{black}}}

\newcommand{\Ns}{N_{\text{s}}}
\newcommand{\NRB}{N_{\text{RB}}}
\newcommand{\NRE}{N_{\text{RE}}}
\newcommand{\NTBS}{N_{\text{TBS}}}

\newcommand{\Nsymb}{N_{\text{symb}}}

\newcommand{\Nt}{N_{\text{T}}}

\newcommand{\Nr}{N_{\text{R}}}

\newcommand{\snreesm}{\text{SNR}_{\text{EESM}}}
\newcommand{\snrxesm}{\text{SNR}_{\text{xESM}}}
\newcommand{\snrth}{\gamma_{\text{th}}}
\newcommand{\bleraw}{p_e}
\newcommand{\Reff}{R_{\text{eff}}}
\newcommand{\Ptotal}{P_{\text{total}}}
\newcommand{\Pstatic}{P_{\text{static}}}

\newcommand{\Pmax}{P_{\text{max}}}

\DeclareMathOperator*{\argmax}{arg\,max}

\newcommand{\SNR}{\text{SNR}}

\newcommand{\figref}[1]{Fig.~\ref{#1}}
\newcommand{\tabref}[1]{Table~\ref{#1}}
\newcommand{\secref}[1]{Section~\ref{#1}}
\newcommand{\algref}[1]{Algorithm~\ref{#1}}

\def\bb0{{\mathbb{0}}}

\def\bb{{\mathbf{b}}}

\def\bp{{\mathbf{p}}}

\def\bs{{\mathbf{s}}}

\def\bv{{\mathbf{v}}}

\def\by{{\mathbf{y}}}

\def\b0{{\mathbf{0}}}

\def\bF{{\mathbf{F}}}

\def\bH{{\mathbf{H}}}
\def\bI{{\mathbf{I}}}

\def\bU{{\mathbf{U}}}
\def\bV{{\mathbf{V}}}

\def\bbC{{\mathbb{C}}}

\def\bbE{{\mathbb{E}}}

\def\bbR{{\mathbb{R}}}

\def\cA{\mathcal{A}}

\def\cI{\mathcal{I}}

\def\cM{\mathcal{M}}
\def\cN{\mathcal{N}}

\def\cS{\mathcal{S}}

\def\sfH{\mathsf{H}}

\def\sfT{\mathsf{T}}

\def\sf0{{\mathsf{0}}}

\def\rm0{{\mathrm{0}}}

\newtheorem{theorem}{Theorem}
\theoremstyle{definition}

\begin{document}
\bstctlcite{IEEEexample:BSTcontrol}
\title{Energy Efficiency Optimization for 5G NR PDSCH:\\ A Cross-Layer Link Abstraction Framework}

\author{\IEEEauthorblockN{Jaebum Park$^{*}$,~Chan-Byoung Chae$^{\dagger}$, and Robert W. Heath Jr.$^{*}$}
\IEEEauthorblockA{$^{*}$Department of Electrical and Computer Engineering, UC San Diego, La Jolla, CA 92093, USA \\
$^{\dagger}$School of Integrated Technology, Yonsei University, Seoul 03722, South Korea\\E-mail: \{jp261, rwheathjr\}@ucsd.edu, cbchae@yonsei.ac.kr }}
\maketitle

\begin{abstract}
Energy efficiency (EE) optimization for cellular links often assumes Shannon-capacity transmission. The 3GPP New Radio (NR) physical downlink shared channel (PDSCH), however, delivers a transport block whose size is fixed by the modulation and coding scheme (MCS), the rank, and the resource allocation chosen by the scheduler. The transport block must also meet a target block error rate (BLER). In this paper, we optimize EE for the NR PDSCH with the standard transport block size as the throughput and the target BLER as the constraint. Exponential effective signal-to-noise ratio (SNR) mapping (EESM) turns the BLER constraint into a convex per-subcarrier power constraint. We derive a closed-form waterfilling-like power allocation and an EE-maximizing MCS selection rule. With a single calibration parameter per MCS, fitted offline over 3GPP tapped delay line C (TDL-C) channels, EESM predicts the BLER without per-realization link-level simulation. Numerical results show that a Shannon-capacity baseline overestimates the EE of equal-power transmission by 1.5 to 1.7 times at high SNR within the transmit power budget.\\
\end{abstract}


\section{Introduction}
The base station consumes the dominant share of mobile network energy, and 5G NR densification amplifies this share~\cite{GreenFutureNetworks2021}. 3GPP NR base stations transmit under discrete modulation, code rate, and target BLER constraints on the physical downlink shared channel (PDSCH)~\cite{3gppTS38.214}, which produce a stepwise relationship between transmit power and effective throughput. Most energy efficiency studies, however, assume Shannon-capacity transmission, and this continuous-rate assumption inflates the predicted energy efficiency of deployed systems~\cite{FengEtAlSurveyEnergyefficientWireless2013a, 3gpp:6GWorkshopReport2025}.

Prior work on energy efficiency optimization for cellular links has predominantly relied on Shannon capacity, with gains shown in orthogonal frequency-division multiple access (OFDMA) systems with large antenna arrays~\cite{NgEtAlEnergyEfficientResourceAllocation2012} and in frequency-selective link adaptation with circuit power awareness~\cite{MiaoEtAlEnergyefficientLinkAdaptation2010}. Our earlier work addressed cross-layer multiple-input multiple-output (MIMO) adaptation~\cite{KimEtAlCrosslayerApproachEnergy2009, ChaeEtAlAdaptiveMIMOTransmission2010} and adaptive waveform and mode switching for power amplifier (PA) efficiency~\cite{ParkEtAlSwitchDFTAdaptiveWaveform2026}. Both lines of work used Shannon rates. Fractional energy efficiency objectives are classically handled through Dinkelbach's parametric approach \cite{IshedenEtAlFrameworkLinkLevelEnergy2012}. None of these formulations capture the discrete MCS structure and BLER constraints of 3GPP NR transmission. Under a target BLER, the throughput of a fixed MCS and rank is a constant, so the fractional objective reduces to transmit-power minimization.

Link-level simulation reproduces the BLER of a 3GPP NR transport block but gives no expression for the BLER as a function of the per-subcarrier power that an optimization can constrain. Effective SNR mapping (ESM) compresses the per-subcarrier SNR vector into a scalar equivalent SNR that predicts BLER through precomputed additive white Gaussian noise (AWGN) reference curves~\cite{HeathLozanoFoundationsMIMOCommunications2018}. EESM \cite{BlankenshipEtAlLinkErrorPrediction2004a} and the mutual-information-based ESM (MIESM) \cite{BrueninghausEtAlLinkPerformanceModels2005} support system-level simulation and link adaptation \cite{FrancisMehtaEESMBasedLinkAdaptation2014}. The compression function of MIESM has no closed-form expression. Learning-based approaches that optimize MCS and power jointly \cite{ParsaEtAlJointPowerAllocation2022} require training and give no analytic solution. The exponential form of EESM makes the BLER constraint of a discrete-MCS link convex in the per-subcarrier power. This convexity gives a closed-form power allocation and an analytic MCS rule under the discrete rates and the target BLER of the standard, which Shannon-based EE formulations do not model.

In this paper, we optimize the EE of the NR PDSCH with the standard transport block size (TBS) as the throughput and the target BLER as the constraint. The scheduler supplies the MCS candidates, the rank, and the resource allocation. The physical (PHY) layer returns the per-subcarrier power and the EE-maximizing MCS. First, we formulate EE maximization under EESM link abstraction and show that, for a fixed MCS and rank, it reduces to a convex transmit-power minimization. Second, we derive a closed-form waterfilling-like power allocation for single-stream transmission, where the water level is set by the EESM threshold and the calibration parameter rather than by a power budget. The multi-stream case is solved numerically with a calibration parameter fitted per rank. Third, we fit one calibration parameter per MCS offline over 3GPP NR tapped delay line C (TDL-C) channels and validate the EESM-predicted BLER against link-level simulation. With equal power across subcarriers, a Shannon-capacity baseline overestimates the single-stream EE by 1.5 to 1.7 times within the transmit power budget. A continuous-rate rank comparison under a radio frequency (RF) power model shows that single-stream transmission stays EE-optimal up to a traffic-demand threshold. The journal version of this paper~\cite{ParkEtAlEnergyEfficiencyOptimization2026} extends the single-link framework to distributed MIMO with baseband and fronthaul power and with traffic-aware mode selection.

\section{System model}
\label{sec:system}
We consider a 3GPP NR downlink PDSCH link between a base station with $\Nt$ transmit antennas and a user equipment with $\Nr$ receive antennas. The maximum number of spatial streams is $\Ns = \min(\Nr, \Nt)$. Throughout this paper, $\subc \in \{0, 1, \ldots, \bandK-1\}$ indexes the $\bandK$ data subcarriers and $j$ indexes a spatial stream. Bold lowercase and uppercase letters denote vectors and matrices, and $(\cdot)^{\sfT}$ and $(\cdot)^{\sfH}$ denote the transpose and the conjugate transpose. We assume perfect channel state information at the base station, obtained through uplink sounding under time-division duplex reciprocity, to isolate the effect of the MCS and BLER constraints.

\subsection{Signal and channel model}
The base station computes the singular value decomposition (SVD) of the frequency-domain channel matrix $\bH[\subc] \in \bbC^{\Nr \times \Nt}$ as
\begin{align}
    \bH[\subc] = \bU[\subc] \mathbf{\Sigma}[\subc] \bV[\subc]^{\sfH},
    \label{eq:svd}
\end{align}
where $\bU[\subc]$ and $\bV[\subc]$ are unitary and $\mathbf{\Sigma}[\subc] \in \bbR^{\Nr \times \Nt}$ contains the singular values $\sigma_0[\subc] \geq \sigma_1[\subc] \geq \cdots \geq \sigma_{\Ns-1}[\subc] \geq 0$. For a transmission rank $r \in \{1, \ldots, \Ns\}$, the base station sets the precoder $\bF[\subc] \in \bbC^{\Nt \times r}$ to the first $r$ columns of $\bV[\subc]$, so $\bF[\subc]$ has orthonormal columns. The data symbol vector $\bs[\subc] \in \bbC^{r \times 1}$ satisfies $\bbE[\bs[\subc]\bs[\subc]^{\sfH}] = \bI_r$. Let $P[\subc] \geq 0$ denote the transmit power at subcarrier $\subc$, shared equally across the $r$ active streams. The vector $\bp = [P[0], \ldots, P[\bandK-1]]^{\sfT}$ collects these powers. The noise $\bv[\subc] \sim \cN_{\bbC}(\b0, N_0 \Delta f \, \bI_{\Nr})$ is independent across subcarriers, where $N_0$ is the noise power spectral density and $\Delta f$ is the subcarrier spacing. The received signal is
\begin{equation}
    \by[\subc] = \sqrt{P[\subc]/r} \, \bH[\subc] \bF[\subc] \bs[\subc] + \bv[\subc].
    \label{eq:rx_signal}
\end{equation}
The receiver combines $\by[\subc]$ with the first $r$ columns of $\bU[\subc]$ to remove inter-stream interference. The per-stream SNR at subcarrier $\subc$ for stream $j$ is~\cite{HeathLozanoFoundationsMIMOCommunications2018}
\begin{equation}
\gamma_j[\subc] = \frac{\sigma_j^2[\subc] P[\subc]}{r \, N_0 \Delta f},
\quad j = 0, 1, \ldots, r-1,
\label{eq:per_stream_snr}
\end{equation}
which is the input to the link abstraction of \secref{sec:eesm}.

\begin{table}[t]
\caption{PDSCH MCS parameters (MCS 4, 10, 19, and 27 \cite{3gppTS38.214})\label{tab:mcs}}
\begin{center}
    {\renewcommand{\arraystretch}{1.1}%
    \small
        \begin{tabular}{c|c|c|c}
        \hline\hline
        $m$ & $Q_m$ & $R_m$ & Spectral efficiency (bits/RE) \\
        \hline
        0 & 2 & 602/1024 & 1.1758 \\
        1 & 4 & 658/1024 & 2.5703 \\
        2 & 6 & 873/1024 & 5.1152 \\
        3 & 8 & 948/1024 & 7.4063 \\
        \hline\hline
        \end{tabular}
    }
\end{center}
\end{table}

\subsection{3GPP NR frame structure and transport block size}
The system uses the 3GPP NR frame structure with numerology $\mu \in \{0, 1, 2\}$ \cite{3gppTS38.211}, which sets the subcarrier spacing $\Delta f = 2^{\mu} \times 15$ kHz and the slot duration $\timeT_{\text{slot}} = 2^{-\mu}$ ms. Each slot contains $\Nsymb = 14$ orthogonal frequency-division multiplexing (OFDM) symbols. Let $Q_m$ denote the modulation order and $R_m$ the code rate for MCS index $m \in \cM$, where $\cM$ is the set of MCS candidates. \tabref{tab:mcs} lists four representative MCS indices, where $m = 0, 1, 2, 3$ denote MCS 4, 10, 19, and 27 of Table 5.1.3.1-2 in TS 38.214~\cite{3gppTS38.214}, the highest-code-rate entry of each modulation order from QPSK to 256-QAM. Four indices keep the offline calibration tractable. The framework applies to any MCS entry once its threshold and calibration parameter are fitted offline.

The TBS follows from the MCS, the number of allocated physical resource blocks (PRBs), and the rank through the count of data resource elements (REs). We denote by $\NRB$ the number of PRBs in the system bandwidth and by $n \in \{1, \ldots, \NRB\}$ the number of allocated PRBs. One demodulation reference signal (DMRS) symbol per slot carries no data in the DMRS pattern of TS 38.211 \cite{3gppTS38.211}. The remaining $\NRE^{\text{PRB}} = 156$ of the 168 REs in a PRB carry data. The TBS for MCS $m$, $n$ allocated PRBs, and rank $r$ is
\begin{equation}
\NTBS(m, n, r) = \NRE^{\text{PRB}} \, R_m \, Q_m \, n \, r,
\label{eq:tbs}
\end{equation}
which is the number of information bits per slot before the TBS quantization of~\cite{3gppTS38.214}. The quantization changes this value by a small fraction, so we treat \eqref{eq:tbs} as the TBS. In the single-user setting, the medium access control (MAC) scheduler allocates the full bandwidth, $n = \NRB$. We write $\NTBS(m, r)$ for $\NTBS(m, \NRB, r)$ throughout.

\section{A cross-layer link abstraction}
\label{sec:eesm}
The link abstraction maps the per-subcarrier SNR to a predicted BLER for each MCS. This mapping gives the optimization of \secref{sec:optimization} a tractable constraint. \figref{fig:block_diagram} shows the proposed cross-layer framework. An offline phase calibrates a per-MCS lookup table. At every slot, an online phase couples the MAC-layer resource allocation with the PHY-layer link abstraction through this table.

\begin{figure}[t]
\centering
\includegraphics[width=0.95\columnwidth]{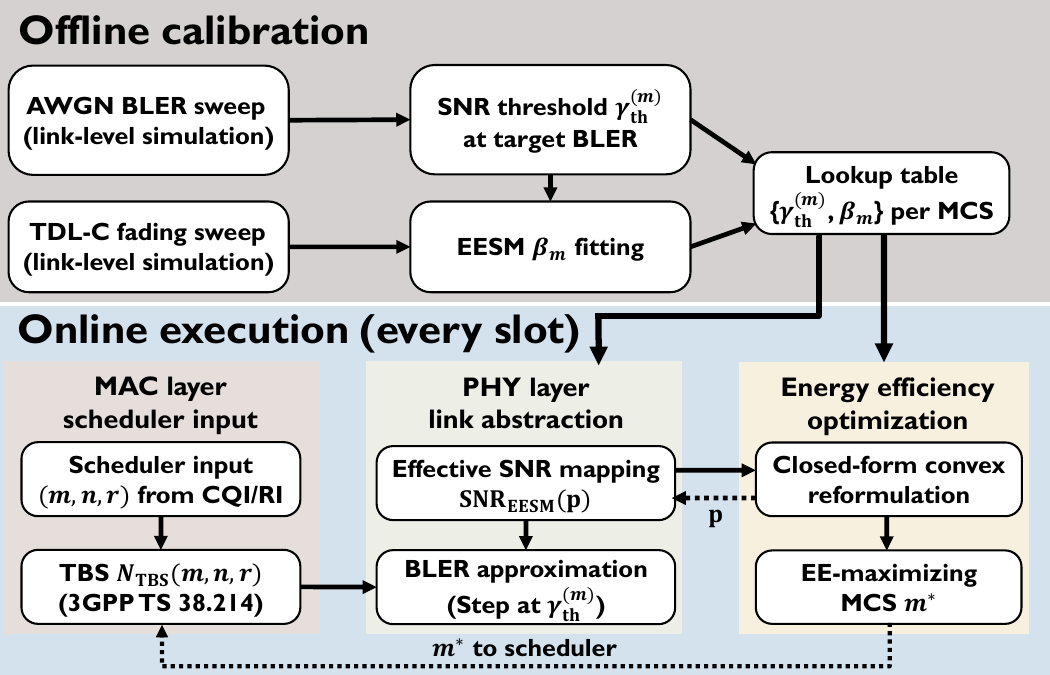}
\caption{Block diagram of the cross-layer link abstraction framework. The offline phase calibrates a per-MCS lookup table $\{\gamma_{\text{th}}^{(m)}, \beta_m\}$ from AWGN BLER references and 3GPP NR TDL-C realizations. At every slot, the MAC scheduler supplies the MCS candidates, the rank, and the PRB allocation from the channel quality indicator (CQI) and rank indicator (RI) reports. The optimization returns the power allocation and the EE-maximizing MCS.}
\label{fig:block_diagram}
\end{figure}

\subsection{Effective SNR mapping framework}
ESM compresses the per-subcarrier SNR vector $(\gamma[0], \ldots, \gamma[\bandK-1])$ into a scalar equivalent SNR. The mapped SNR predicts BLER by matching the frequency-selective channel to an AWGN channel at that SNR. For a single stream, let $\gamma[\subc]$ denote the SNR at subcarrier $\subc$. All ESM variants use the form
\begin{equation}
\snrxesm \approx \beta \, g^{-1}\!\left( \frac{1}{\bandK} \sum_{\subc=0}^{\bandK-1} g\!\left( \frac{\gamma[\subc]}{\beta} \right) \right),
\label{eq:xesm}
\end{equation}
where $g(\cdot)$ is a monotone mapping function and $\beta > 0$ is a calibration parameter. EESM uses $g(\gamma) = e^{-\gamma}$~\cite{BlankenshipEtAlLinkErrorPrediction2004a}, so the BLER constraint is an explicit sum of exponentials in the per-subcarrier power. MIESM uses the constellation-constrained mutual information, whose lack of a closed-form expression rules out a closed-form allocation. We adopt EESM for this tractability, although MIESM predicts BLER more accurately~\cite{BrueninghausEtAlLinkPerformanceModels2005}.

\subsection{EESM derivation}
The original EESM averages $e^{-\gamma[\subc]/\beta}$ over the $\bandK$ subcarriers that carry a single-stream codeword and inverts the average through the logarithm~\cite{BlankenshipEtAlLinkErrorPrediction2004a}. For $r > 1$, the combining of Section~\ref{sec:system} leaves $r$ interference-free streams. 3GPP NR maps one transport block onto up to four layers as a single codeword \cite{3gppTS38.211}, so the transport block succeeds or fails jointly across all $r$ streams. The codeword therefore spans the $r\bandK$ sub-channels indexed by $(j, \subc)$. With $a_{j,\subc} = \sigma_j^2[\subc]/(\beta r N_0 \Delta f)$, the EESM over these sub-channels gives the effective SNR as a function of $\bp$,
\begin{align}
\snreesm(\bp) = -\beta \ln\!\left( \frac{1}{r\bandK} \sum_{j=0}^{r-1} \sum_{\subc=0}^{\bandK-1} e^{-a_{j,\subc} P[\subc]} \right).
\label{eq:eesm_multistream}
\end{align}
Let $p_e(\cdot)$ denote the AWGN BLER curve of MCS $m$ and $\snrth^{(m)}$ the SNR at which it equals the target BLER $p_e^{\text{target}}$. The AWGN BLER curves of the NR low-density parity-check (LDPC) code in Fig.~\ref{fig:eesm_calib} have a waterfall transition within approximately 1 dB for long transport blocks, whereas adjacent thresholds $\snrth^{(m)}$ are several dB apart. We therefore approximate the BLER by the step function
\begin{equation}
\bleraw(\SNR) \approx
\begin{cases}
1, & \SNR < \snrth^{(m)},\\
0, & \SNR \geq \snrth^{(m)},
\end{cases}
\label{eq:step_bler}
\end{equation}
which makes the effective throughput a deterministic function of $\snreesm(\bp)$. The approximation loses accuracy when the target lies in the error-floor region or when adjacent thresholds are closer than the waterfall width, as with a denser MCS set.

\subsection{Offline EESM calibration}
The calibration parameter $\beta$ in \eqref{eq:eesm_multistream} is fitted offline for each MCS index $m$ by \algref{alg:calibration}, following \cite{LagenEtAlNewRadioPhysical2020}. The algorithm runs single-stream link-level simulation over TDL-C realizations on an SNR grid with 1 dB steps, and its BLER window excludes the error-floor and saturation regions. For $r > 1$, repeating the algorithm with the target $r$ in \eqref{eq:eesm_multistream} gives a per-rank parameter $\beta_m^{(r)}$. The calibrated $\beta_m$ is insensitive to the channel realization and depends on the MCS, consistent with prior EESM link-error prediction~\cite{DonthiMehtaAccurateModelEESM2011}. This paper adds its use inside the optimization constraint of Section~\ref{sec:optimization}, where a single scalar per MCS preserves the convex structure.

\begin{algorithm}[t]
\caption{Offline EESM calibration for MCS index $m$}
\label{alg:calibration}
\begin{algorithmic}[1]
\STATE \textbf{Input:} MCS index $m$, SNR grid $\{\gamma_{\SNR}^{(i)}\}_{i=1}^{N_{\SNR}}$,\\ channel realizations per SNR point $N_{\text{cal}}$
\STATE \textbf{Output:} Calibrated parameter $\beta_m$
\STATE Generate AWGN BLER reference curve $\bleraw(\gamma)$ for $m$
\STATE Initialize the set of per-realization fits $\cS_m \gets \emptyset$
\FOR{$i = 1$ to $N_{\SNR}$}
    \STATE Set operating SNR $\gamma_{\SNR} \gets \gamma_{\SNR}^{(i)}$
    \STATE Run link-level simulation over the TDL-C channel to obtain $\text{BLER}^{(i)}$ at $\gamma_{\SNR}$
    \IF{$\text{BLER}^{(i)} \in [10^{-3}, 5 \times 10^{-2}]$}
        \STATE Find $\gamma^{(i)}_{\text{target}}$ such that $\bleraw(\gamma^{(i)}_{\text{target}}) = \text{BLER}^{(i)}$
        \FOR{$\ell = 1$ to $N_{\text{cal}}$}
            \STATE Generate TDL-C channel realization $\bH^{(\ell)}$ and compute $\{\gamma[\subc]^{(\ell)}\}_{\subc=0}^{\bandK-1}$ at $\gamma_{\SNR}$
            \STATE Solve for the $\beta^{(i,\ell)}$ satisfying $-\beta^{(i,\ell)} \ln\!\left( \frac{1}{\bandK} \sum_{\subc=0}^{\bandK-1} e^{-\gamma[\subc]^{(\ell)}/\beta^{(i,\ell)}} \right) = \gamma^{(i)}_{\text{target}}$
            \STATE $\cS_m \gets \cS_m \cup \{\beta^{(i,\ell)}\}$
        \ENDFOR
    \ENDIF
\ENDFOR
\STATE \textbf{return} $\beta_m \gets \frac{1}{|\cS_m|} \sum_{\beta \in \cS_m} \beta$
\end{algorithmic}
\end{algorithm}

\section{Energy-efficient optimization}
\label{sec:optimization}
We formulate EE maximization as a cross-layer optimization. The MAC-layer TBS sets the numerator of the EE and the PHY-layer power sets the denominator. The EESM constraint ties the two through $\snrth^{(m)}$ and $\beta_m$. Since the power minimization for each MCS candidate has a closed-form solution for $r = 1$ and a scalar root-finding solution for $r > 1$, the selection needs neither link-level simulation nor a learned policy.

\subsection{Power model and energy efficiency metric}
The total power decomposes into a static component and a transmit-power component as
\begin{equation}
\Ptotal(\bp) = \Pstatic + \frac{1}{\eta_{\text{PA}}} \mathbf{1}^{\sfT} \bp,
\label{eq:power_model}
\end{equation}
where $\Pstatic$ aggregates the baseband, RF chain, and idle-hardware contributions, $\mathbf{1}$ is the all-ones vector, and $\eta_{\text{PA}} \in (0,1]$ is the PA efficiency. Baseband and fronthaul terms that scale with the rank or the PRB allocation are constants for a fixed $(m, r, n)$~\cite{AuerEtAlHowMuchEnergy2011}. They leave the power allocation unchanged and only shift the rank crossover of \secref{sec:crossover}. The effective throughput is
\begin{equation}
\Reff(m, r, \bp) = (1 - \bleraw(\snreesm(\bp))) \, \frac{\NTBS(m, r)}{\timeT_{\text{slot}}}.
\label{eq:reff}
\end{equation}
The energy efficiency is 
\begin{equation}
\eta_{\text{EE}}(m, r, \bp) = \frac{\Reff(m, r, \bp)}{\Ptotal(\bp)},
\label{eq:ee_metric}
\end{equation}
which the optimization maximizes over the MCS index $m$ and the power vector $\bp$ for the rank $r$ given by the scheduler.

\subsection{Convex problem formulation}
Under the step-function approximation \eqref{eq:step_bler}, the effective throughput equals $\NTBS(m,r)/\timeT_{\text{slot}}$ when $\snreesm(\bp) \geq \snrth^{(m)}$ and zero otherwise. For a fixed $(m, r)$, this throughput is constant and $\Pstatic$ is a fixed offset. The radiated power at subcarrier $\subc$ equals $P[\subc]$ because $\bF[\subc]$ has orthonormal columns, so the PA budget $\Pmax$ bounds $\mathbf{1}^{\sfT} \bp$. Maximizing $\eta_{\text{EE}}$ therefore reduces to the power minimization
\begin{align}
\min_{\bp} \quad & \mathbf{1}^{\sfT} \bp, \label{eq:p1_obj}\\
\text{s.t.} \quad & \snreesm(\bp) \geq \snrth^{(m)}, \label{eq:p1_eesm}\\
& \mathbf{1}^{\sfT} \bp \leq \Pmax, \label{eq:p1_pmax}\\
& \bp \succeq \mathbf{0}. \label{eq:p1_nonneg}
\end{align}
We substitute \eqref{eq:eesm_multistream} into \eqref{eq:p1_eesm}, multiply by $-1/\beta_m$, and apply the exponential function to obtain the equivalent constraint
\begin{equation}
\frac{1}{r\bandK} \sum_{j=0}^{r-1} \sum_{\subc=0}^{\bandK-1} e^{-a_{j,\subc} P[\subc]} \leq e^{-\snrth^{(m)}/\beta_m}.
\label{eq:eesm_convex}
\end{equation}
The objective is linear. The constraint \eqref{eq:eesm_convex} is a sublevel set of a sum of exponentials, each convex in $P[\subc]$. The problem is therefore convex, and the Karush-Kuhn-Tucker (KKT) conditions give the global optimum.

\subsection{Closed-form power allocation}
Let $\lambda \geq 0$, $\nu \geq 0$, and $\xi_{\subc} \geq 0$ be the multipliers for the EESM, total-power, and non-negativity constraints. With $a_{j,\subc}$ evaluated at $\beta = \beta_m$, the stationarity condition is
\begin{equation}
1 + \nu - \frac{\lambda}{r\bandK} \sum_{j=0}^{r-1} a_{j,\subc} e^{-a_{j,\subc} P[\subc]} - \xi_{\subc} = 0.
\label{eq:lagrangian}
\end{equation}
For $r = 1$, the sum in \eqref{eq:lagrangian} reduces to a single exponential, and Theorem~\ref{thm:r1} gives the solution in closed form. For $r > 1$, \eqref{eq:lagrangian} is transcendental in $P[\subc]$ and is solved by a numerical root-finder.

\begin{theorem}
\label{thm:r1}
Let $\cA = \{\subc : P^*[\subc] > 0\}$ denote the active set and $\cI = \{0, \ldots, \bandK-1\} \setminus \cA$ the inactive set. For $r = 1$ and an MCS $m$ whose minimum transmit power satisfies the budget $\Pmax$, the optimal power allocation for $\subc \in \cA$ is
\begin{equation}
P^*[\subc] = \frac{\beta_m N_0 \Delta f}{\sigma_0^2[\subc]} \ln\!\left( \frac{\lambda_m \sigma_0^2[\subc]}{\bandK \, \beta_m N_0 \Delta f} \right),
\label{eq:p_star}
\end{equation}
where $\lambda_m > 0$, the multiplier $\lambda$ of \eqref{eq:lagrangian} for MCS $m$, satisfies the EESM constraint with equality on $\cA$ and is given by
\begin{equation}
\lambda_m = \left( \sum_{\subc \in \cA} \frac{1}{\sigma_0^2[\subc]} \right) \frac{\beta_m N_0 \Delta f}{e^{-\snrth^{(m)}/\beta_m} - |\cI|/\bandK}.
\label{eq:lambda_star}
\end{equation}
\end{theorem}

\begin{IEEEproof}
We set $r = 1$ and $a_{\subc} = \sigma_0^2[\subc]/(\beta_m N_0 \Delta f)$ in \eqref{eq:lagrangian}. The budget constraint is inactive for such an MCS, so $\nu = 0$. Complementary slackness gives $\xi_{\subc} = 0$ for $\subc \in \cA$, which yields $e^{-a_{\subc} P^*[\subc]} = \bandK/(\lambda_m a_{\subc})$. Taking the logarithm produces \eqref{eq:p_star}. We substitute back into the EESM equality constraint and solve for $\lambda_m$ to obtain \eqref{eq:lambda_star}. The multiplier satisfies $\lambda_m > 0$, since $\lambda_m = 0$ forces $P^*[\subc] = 0$ for all $\subc$ through \eqref{eq:lagrangian}, in which case the left side of \eqref{eq:eesm_convex} equals one and exceeds the right side. Subcarriers with $\lambda_m a_{\subc}/\bandK \leq 1$ are inactive, which yields the active set partition in \eqref{eq:p_star}.
\end{IEEEproof}

The closed form in \eqref{eq:p_star} has a waterfilling-like structure, and an iterative drop procedure finds its active set. The procedure starts from $\cA = \{0, \ldots, \bandK - 1\}$, computes $\lambda_m$ from \eqref{eq:lambda_star}, evaluates \eqref{eq:p_star}, and removes the subcarriers with $P^*[\subc] \leq 0$ until the set no longer changes. Classical waterfilling sets the water level as the dual variable of a power budget. The multiplier $\lambda_m$ is instead the dual variable of the EESM threshold constraint, so it depends on $\snrth^{(m)}$ and $\beta_m$. Algorithm~\ref{alg:joint} gives the online phase of Fig.~\ref{fig:block_diagram}. It evaluates the closed form for $r = 1$ and the numerical solution for $r > 1$, and it returns the MCS with the highest EE.

\begin{algorithm}[t]
\caption{EE-maximizing MCS selection and power allocation}
\label{alg:joint}
\begin{algorithmic}[1]
\STATE \textbf{Input:} rank $r$, MCS candidates $\cM$, singular values $\{\sigma_j^2[\subc]\}$, lookup table $\{\snrth^{(m)}, \beta_m\}$, noise $N_0 \Delta f$, budget $\Pmax$
\STATE \textbf{Output:} MCS $m^*$ and power $\{P^*[\subc]\}$
\FOR{each $m \in \cM$}
\IF{$r = 1$}
\STATE Compute $\lambda_m$ from \eqref{eq:lambda_star} and $P^{(m)}[\subc]$ from \eqref{eq:p_star}
\ELSE
\STATE Numerically solve \eqref{eq:lagrangian} with the EESM equality constraint to obtain $P^{(m)}[\subc]$
\ENDIF
\STATE Compute $\eta_{\text{EE}}^{(m)}$ from \eqref{eq:ee_metric} and set it to zero if $\sum_{\subc} P^{(m)}[\subc] > \Pmax$
\ENDFOR
\STATE $m^* = \argmax_{m} \eta_{\text{EE}}^{(m)}$, \quad $P^*[\subc] = P^{(m^*)}[\subc]$
\end{algorithmic}
\end{algorithm}

For $r = 1$, each pass of the drop procedure costs $O(\bandK)$ operations and removes at least one subcarrier, so the procedure ends within $|\cI| + 1$ passes and costs $O(\bandK^2)$ in the worst case. For $r > 1$, each iteration of the multiplier search solves \eqref{eq:lagrangian} by scalar root-finding on every subcarrier, which costs $O(\bandK)$ function evaluations. The loop repeats this computation for each of the $|\cM|$ MCS candidates. No link-level simulation runs online, since $\beta_m$ and $\snrth^{(m)}$ come from the offline lookup table.

\section{Simulation results}
\label{sec:simulation}
We evaluate the framework with a 3GPP NR PDSCH link-level simulator. The simulator performs LDPC coding, OFDM modulation, and minimum mean-square-error detection under perfect receive-side channel knowledge. The LDPC decoder uses base graph 1 with layered offset min-sum scheduling and at most 10 iterations, so the calibrated $\beta_m$ is specific to this receiver. Each TDL-C operating point runs up to $10^5$ slots and stops after 1000 block errors. Every slot draws an independent channel realization with a maximum Doppler shift of 9 Hz. The per-subcarrier gains of $N_{\text{cal}} = 100$ realizations per SNR point are stored for the calibration. Table~\ref{tab:sim} lists the simulation parameters. \secref{sec:calib} and \secref{sec:ee_snr} evaluate a single-input single-output (SISO) link, for which Theorem~\ref{thm:r1} applies. \secref{sec:crossover} compares the two ranks for $\Nt = \Nr = 2$ with an RF-only power model defined there.

\begin{table}[t]
\caption{Simulation parameters. Sections V-A and V-B use $\Nt = \Nr = 1$, and Section V-C uses $\Nt = \Nr = 2$.\label{tab:sim}}
\begin{center}
    {\renewcommand{\arraystretch}{1.1}%
    \small
    \setlength{\tabcolsep}{3pt}%
        \begin{tabular}{c|c|c|c}
        \hline\hline
        Parameter & Value & Parameter & Value \\
        \hline
        $f_c$ & 3.5 GHz & $\Pmax$ & 33 dBm \\
        $\Delta f$ & 30 kHz ($\mu = 1$) & $\eta_{\text{PA}}$ & 0.38 \\
        $\bandK$ & 3276 & $\Pstatic$ & 20 W \\
        $\NRB$ & 273 & $p_e^{\text{target}}$ & $10^{-2}$ \\
        Path loss & 82.9 dB & $N_0$ & $-174$ dBm/Hz \\
        Noise figure & 7 dB & $N_{\text{cal}}$ & 100 \\
        $\timeT_{\text{slot}}$ & 0.5 ms & $N_{\SNR}$ & 20 \\
        Delay spread & 100 ns & & \\
        \hline\hline
        \end{tabular}
    }
\end{center}
\end{table}

\subsection{EESM calibration accuracy}
\label{sec:calib}
A single $\beta_m$ per MCS maps the frequency-selective TDL-C channel to an equivalent AWGN point. Figure~\ref{fig:eesm_calib} compares the EESM-predicted BLER over TDL-C channels with a 100 ns delay spread to the AWGN reference curves of the four MCS indices. For each MCS and SNR point, the stored realizations give $\snreesm$ values with the calibrated $\beta_m$. Within the BLER window of \algref{alg:calibration}, the resulting points align with the AWGN curve of every MCS. The thresholds $\snrth^{(m)}$ at the target BLER are 2.9, 9.1, 18.0, and 25.6 dB, several dB apart as \eqref{eq:step_bler} assumes.

\begin{figure}[t]
\centering
\includegraphics[width=0.85\columnwidth]{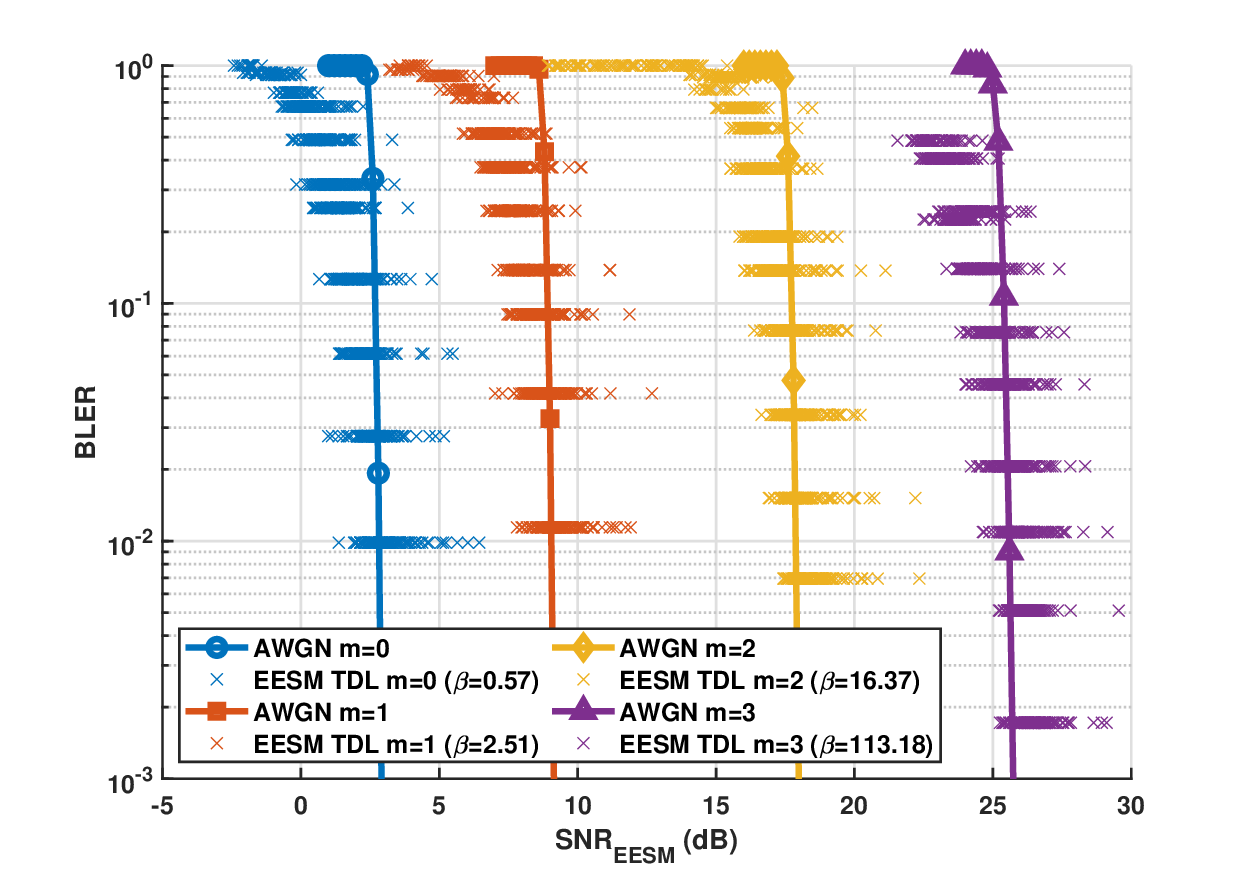}
\caption{AWGN BLER reference curves (lines) and EESM-predicted BLER over TDL-C realizations with 100 ns delay spread (markers) for a SISO link. The calibrated values are $\beta_m = 0.57$, $2.51$, $16.37$, and $113.18$ for $m = 0, 1, 2, 3$.}
\label{fig:eesm_calib}
\end{figure}

\subsection{Energy efficiency versus SNR per MCS}
\label{sec:ee_snr}
The EE of each MCS peaks at the lowest operating SNR that meets its EESM threshold. Figure~\ref{fig:ee_snr} shows the per-MCS EE versus the TDL-C operating SNR, which sets the transmit power through the link budget of Table~\ref{tab:sim}, with the urban macro line-of-sight path loss at 100 m. The EE stays at zero while the BLER exceeds the target and jumps to its peak once the effective SNR, averaged over the $N_{\text{cal}}$ realizations at that operating SNR, reaches $\snrth^{(m)}$. Frequency-selective fading adds a margin over the AWGN threshold, so QPSK turns on near 11 dB, above its 2.9 dB threshold. The EE then decays as the transmit power in the denominator of \eqref{eq:ee_metric} grows while the numerator stays fixed. A higher MCS raises the numerator but also raises $\snrth^{(m)}$ and the required power, so the EE-maximizing MCS depends on the operating SNR. Selecting that MCS at each SNR traces the upper envelope of the four curves. The sweep uses equal power across subcarriers. The allocation of Theorem~\ref{thm:r1} meets each threshold with no more power, so the curves lower-bound the EE of the framework.

The Shannon-capacity baseline is the continuous-rate model optimized in~\cite{NgEtAlEnergyEfficientResourceAllocation2012, MiaoEtAlEnergyefficientLinkAdaptation2010}. Since the power model is otherwise identical, the comparison isolates the effect of the MCS and BLER constraints. The baseline replaces the EESM-based throughput with $\NRE^{\text{PRB}} \, \NRB \, r \, \log_2(1 + \SNR) / \timeT_{\text{slot}}$ at the operating SNR. From 33 dB to 37 dB, where the transmit power reaches the budget $\Pmax$, the baseline overestimates the EESM envelope by 1.5 to 1.7 times. Since the highest MCS of the table is active in this range, the gap does not depend on the choice of the four indices. Beyond the budget, the ratio reaches 2.0 at 45 dB. The EE is zero below the turn-on point of $m = 0$ because the lower-rate QPSK entries of the table are not simulated.
\begin{figure}[t]
\centering
\includegraphics[width=0.85\columnwidth]{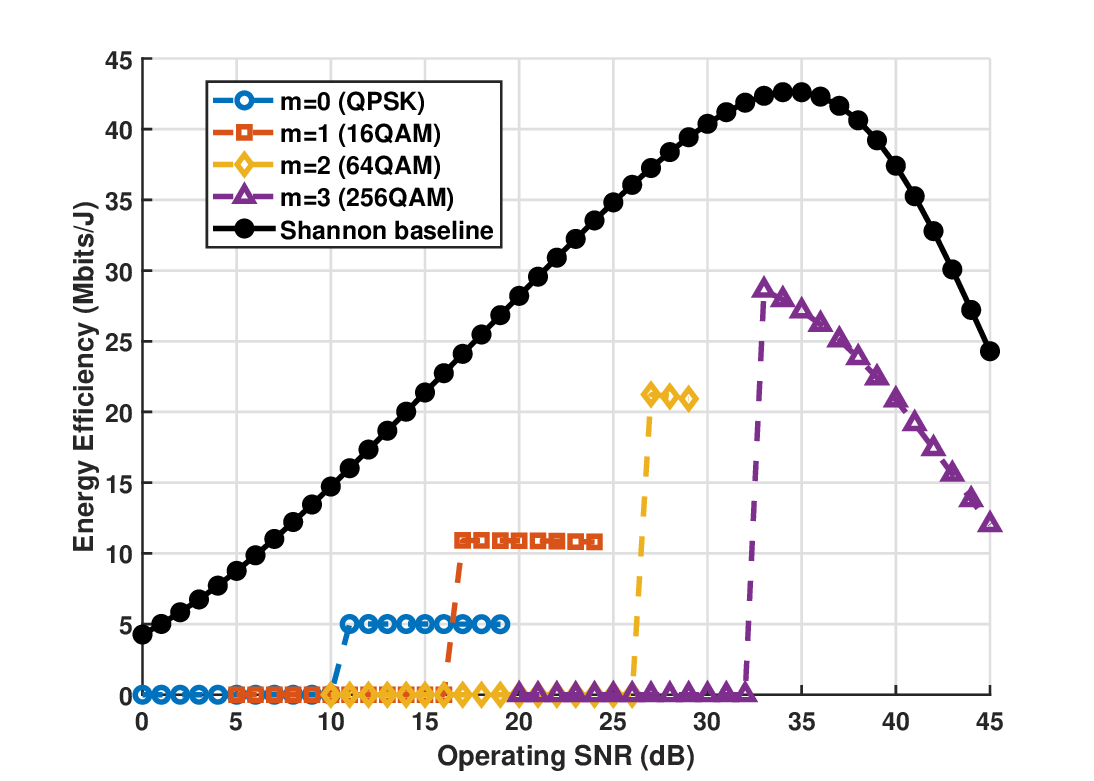}
\caption{EE versus operating SNR for the four MCS indices on a SISO link over a TDL-C channel with $\NRB = 273$. The Shannon-capacity baseline assumes continuous-rate transmission under the same power model.}
\label{fig:ee_snr}
\end{figure}

\subsection{Single-stream versus dual-stream crossover}
\label{sec:crossover}
The power model determines which rank the scheduler should supply at a given load. Let $\Lambda$ denote the traffic demand as a spectral efficiency in bps/Hz over the system bandwidth. The RF-only model keeps the PA term of \eqref{eq:power_model} and replaces $\Pstatic$ with the static power of the active RF chains, $P_{\text{RF},0}^{(r)}$ for rank $r$, at 0.5 W per chain. The example takes a frequency-flat channel $\bH$ with $\bH^{\sfH} \bH = \mathrm{diag}(\sigma_0^2, \sigma_1^2)$, for which the rank-1 precoder of \secref{sec:system} selects one antenna. Single-stream transmission ($r = 1$) then drives one RF chain and dual-stream transmission ($r = 2$) drives both, so $P_{\text{RF},0}^{(1)} = 0.5$ W and $P_{\text{RF},0}^{(2)} = 1.0$ W. Let $P_{\text{tx}}^{(r)}(\Lambda)$ denote the minimum transmit power that meets demand $\Lambda$ at rank $r$. For a fixed demand, the EE-optimal rank is the one with the lower total power. The crossover $\Lambda^*$ solves
\begin{equation}
P_{\text{RF},0}^{(2)} + \frac{P_{\text{tx}}^{(2)}(\Lambda^*)}{\eta_{\text{PA}}} = P_{\text{RF},0}^{(1)} + \frac{P_{\text{tx}}^{(1)}(\Lambda^*)}{\eta_{\text{PA}}}.
\label{eq:crossover}
\end{equation}

Figure~\ref{fig:crossover} illustrates \eqref{eq:crossover} with continuous-rate expressions for $P_{\text{tx}}^{(r)}(\Lambda)$ and stream gains $\sigma_0^2 = 1.5$ and $\sigma_1^2 = 0.5$. In this illustration, single-stream transmission is more power-efficient up to $\Lambda^* = 4.93$ bps/Hz. Below $\Lambda^*$, single-stream transmission consumes less power because of its lower static RF power. Above $\Lambda^*$, dual-stream transmission spreads the demand over two streams and lowers the per-stream transmit power. Baseband terms that scale with the rank add to the left side of \eqref{eq:crossover} and extend the single-stream range.

\begin{figure}[t]
\centering
\includegraphics[width=0.85\columnwidth]{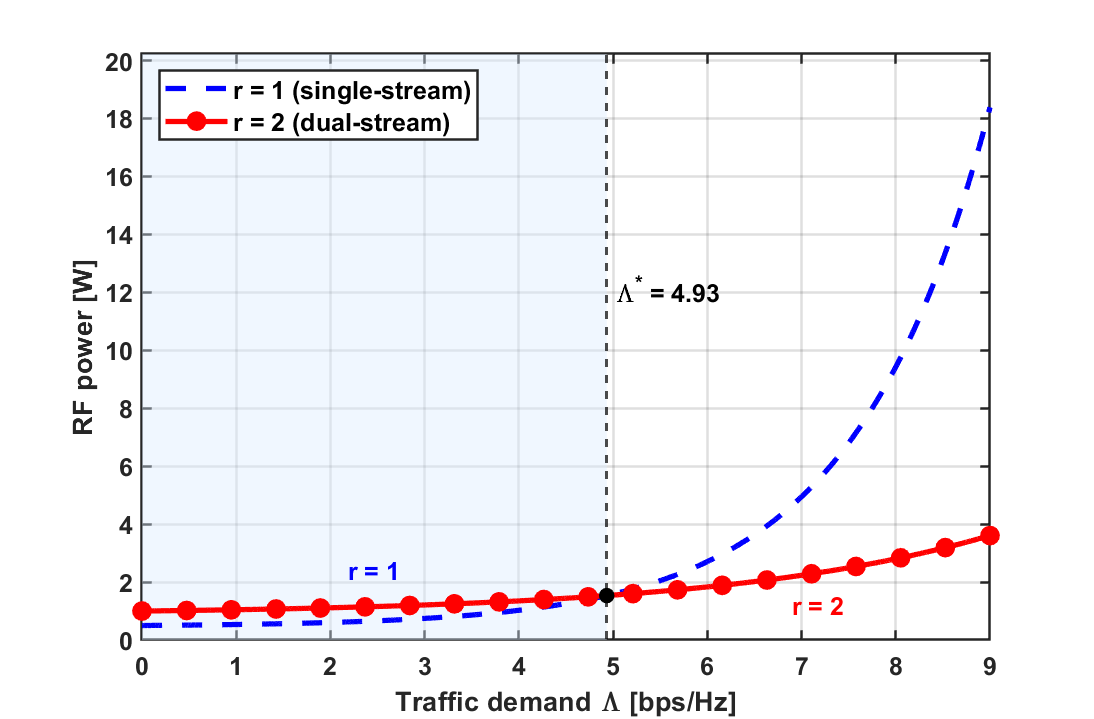}
\caption{RF power versus traffic demand $\Lambda$ for single-stream and dual-stream transmission with $\Nt = \Nr = 2$ under the RF-only model. The curves cross at $\Lambda^* = 4.93$ bps/Hz, and the shaded region marks the demand range where single-stream transmission uses less power.}
\label{fig:crossover}
\end{figure}

\section{Conclusion}
\label{sec:conclusion}
In this paper, we proposed a closed-form energy-efficient power allocation for the NR PDSCH under EESM-based link abstraction, with the water level set by the EESM threshold and one calibration parameter per MCS. A Shannon-capacity baseline overestimated the EE of equal-power transmission by 1.5 to 1.7 times at high SNR within the transmit power budget. Under an RF-only power model with continuous-rate transmit power, single-stream transmission was more power-efficient up to a traffic demand of 4.93 bps/Hz. We assumed perfect channel state information, a single-user link, and full PRB allocation to isolate the effect of the MCS and BLER constraints. We expect channel estimation errors and interference to raise the transmit power that meets the EESM threshold. The framework depends on the number of streams, so larger arrays enter only through the singular values. Future work will address multi-user and imperfect channel state settings and will compare the framework with MIESM-based link adaptation.

\section*{Acknowledgement}
This work was supported in part by the National Science Foundation under Grant Nos. 2414678 and 2549148, and by the Korean government (RS-2024-00428780, RS-2024-00397216).

\bibliographystyle{IEEEtran}
\bibliography{JP_EE_ICC}

\end{document}